\documentclass{llncs}[11pt]
\usepackage{amssymb,amsmath}
\usepackage{graphicx,color}
\usepackage{textcomp}
\usepackage{wasysym}
\begin{document}
%
\frontmatter          
\pagestyle{headings}  

\mainmatter   
%
%
\title{Approximating the Cubicity of Trees
}
%

\author{Jasine Babu \inst{1} \and Manu Basavaraju \inst{2} \and L Sunil Chandran \inst{1} \and Deepak Rajendraprasad \inst{3} \and Naveen Sivadasan \inst{4}}
\authorrunning{Babu et al.} 
\institute{Indian Institute of Science, Bangalore, India.
\and Department of Informatics, University of Bergen, Norway,
\and University of Haifa, Israel,
\and Indian Institute of Technology, Hyderabad, India.\\
\email{jasine@csa.iisc.ernet.in, iammanu@gmail.com, sunil@csa.iisc.ernet.in, deepakmail@gmail.com, nsivadasan@iith.ac.in}}
\maketitle  

%
%
 

\setlength{\parindent}{1cm}
\begin{abstract}
Cubicity of a graph $G$ is the smallest dimension $d$, for which $G$ is a unit disc graph in ${\mathbb{R}}^d$, under the $l^\infty$ metric, i.e. 
$G$ can be represented as an intersection graph of $d$-dimensional (axis-parallel) unit hypercubes. 
We call such an intersection representation a $d$-dimensional 
cube representation of $G$.  
Computing cubicity is known to be inapproximable in polynomial time, within an $O(n^{1-\epsilon})$ factor for any $\epsilon >0$, unless $\text{NP}=\text{ZPP}$. 

\parindent 0.5cm In this paper, we present a randomized algorithm that runs in polynomial time and computes cube representations of trees, of dimension 
within a constant factor of the optimum.  It is also shown that the cubicity of trees can be approximated within 
a constant factor in deterministic polynomial time, if the cube representation is not required to be computed. As far as we know, this is the first constant factor approximation 
algorithm for computing the cubicity of trees. It is not yet clear whether computing the cubicity of trees is NP-hard or not. 
\keywords{Cubicity, Approximation Algorithm}
\end{abstract} 
\section{Introduction}
Cubicity of a graph $G$, denoted by $\operatorname{cub}(G)$ is the smallest dimension $d$ such that $G$ can be represented as an intersection graph 
of $d$-dimensional (axis-parallel) unit hypercubes. 
In other words, $\operatorname{cub}(G)$ is the smallest dimension $d$ for which $G$ is a unit disc graph in ${\mathbb{R}}^d$, under the $l^\infty$ metric. 
It is not difficult to see that, $\operatorname{cub}(G)$ is the smallest integer $d$ such that $G$ can be represented as the intersection 
of $d$ unit interval graphs on the same vertex set $V(G)$; i.e $E(G)=E(I_1) \cap E(I_2) \cap \cdots \cap E(I_d)$, where $I_1, I_2, \ldots, I_d$ 
are unit interval graphs with $V(I_i)=V(G)$, for $1 \le i \le d$ \cite{Rob1}. 
If we relax the requirement of unit interval graphs to interval graphs, the corresponding parameter is called boxicity. 
Equivalently, graphs of boxicity at most $d$ are the intersection graphs of $d$-dimensional axis parallel boxes. 
These parameters were introduced by F. S. Roberts \cite{Rob1} in 1968 for studying some problems in Ecology. 
It is easy to see that $\operatorname{box}(G) \le \operatorname{cub}(G)$. 
Boxicity (resp. cubicity) of a graph on $n$ vertices is at most $\left \lfloor\frac{n}{2} \right \rfloor$ (resp. $\left \lceil\frac{2 n}{3} \right \rceil$) \cite{Rob1}. 
By convention, cubicity and boxicity of a complete graph are zero. It follows from the definitions that $\operatorname{cub}(G) \le 1$, if and only if 
$G$ is a unit interval graph and $box(G) \le 1$, if and only if $G$ is an interval graph. 
It is also known that planar graphs have boxicity at most three and tress have boxicity at most two \cite{Thom1,Chinthan}. 

Since unit interval graphs are polynomial time recognizable, whether $\operatorname{cub}(G)$ $\le 1$ is polynomial time decidable. 
However, deciding whether a graph has cubicity at most $k$ is NP-hard, even for $k=2$ and $k=3$ \cite{Yan1,Breu1998}.
Chalermsook et al. \cite{Chalermsook2013} showed that boxicity and cubicity problems are inapproximable in polynomial time, within an $O(n^{1-\epsilon})$ factor 
for any $\epsilon >0$, unless $\text{NP}=\text{ZPP}$. This hardness result holds for graph classes like bipartite, co-bipartite, and split graphs as well. 

There are not many approximation algorithms known to exist for these problems, even for special classes of graphs. 
As far as we know, an $o(n)$ factor approximation algorithm for computing the cubicity of general graphs \cite{AdigaIPEC} and a constant factor 
approximation algorithm with an extra additive error of $\log{n}$ for computing the cubicity of circular arc graphs \cite{Jas1} are the only non-trivial approximation 
algorithms known for the cubicity problem. 

In this paper, we present a randomized algorithm that runs in polynomial time, for computing cube representations of trees. 
Our algorithm computes cube representations of trees of dimension within a constant factor of the optimum. 
If we do not require a corresponding cube representation, then the cubicity of trees can be approximated within 
a constant factor in polynomial time, without using any randomization. The algorithm presented here seems to be the first constant factor approximation 
algorithm for computing the cubicity of trees. It is not yet clear whether computing the cubicity of trees is NP-hard or not. 

Our randomized procedure borrows its ideas from the randomized algorithm devised by Krauthgamer et al. \cite{Krauth2007}, for approximating 
the intrinsic dimensionality of trees. This parameter is fundamentally different and is incomparable with cubicity in general 
(See Appendix for a detailed comparison between the two parameters). However, 
it comes as a surprise that their proof technique works more or less the same way for cubicity of trees, with some problem specific modifications to handle 
the details and the base cases. This is more surprising because Krauthgamer et al. \cite{Krauth2007} devised an $O(\log \log n)$ factor approximation for 
intrinsic dimensionality of general graphs, by extending the proof techniques used for trees whereas cubicity for general graphs is inapproximable within $O(n^{1-\epsilon})$ factor for any $\epsilon > 0$, unless NP=ZPP.
\section{Preliminaries}
In this paper, we are dealing with only finite graphs, without self loops or multi edges. 
Unless specified otherwise, logarithms are taken to the base $2$. A unit hypercube in ${\mathbb{R}}^d$ is a hypercube whose sides are of unit length in the usual 
Euclidean metric, i.e it is a disc in ${\mathbb{R}}^d$ of radius $\frac{1}{2}$ under the $l^\infty$ metric. 
We use $\| \|_\infty$ to denote the $l^\infty$ norm.
We consider our tress as rooted trees in which the root vertex is considered to be at depth zero and 
for any other vertex, its depth is given by its distance from the root. For any two vertices $u$ and $v$ of a tree $T$, the least common
ancestor of $u$ and $v$ is the vertex with the minimum depth on the path between $u$ and $v$ in $T$. 
If $u, v$ are two vertices in a graph $G$, we use $d_{uv}(G)$ to denote the distance 
between $u$ and $v$ in $G$ and when it clear which graph we are talking out, we just use $d_{uv}$. 
\subsection{Cube representations, embeddings and weight-vector assignments to edges} \label{weightAssignments}
Let $G$ be a graph and suppose $f : V(G) \mapsto {\mathbb{R}}^d$ is such that $\|f(v) - f(u)\|_\infty \le 1$ if and only if $u$ and $v$ are 
adjacent in $G$. If we consider unit hypercube corresponding to a vertex $v$ as the unit hypercube centered at $f(v)$, 
then it is easy to see that the hypercubes corresponding to $u$ and $v$ intersect if and only if $\|f(v) -f(u)\|_{\infty} \le 1$. 
Conversely, given a cube representation of $G$ in $d$ dimensions, for any $v \in V(G)$ we can define $f(v)$ as the vector corresponding to the center of the 
hypercube associated with $v$. Since we derived $f$ from a cube representation of $G$, it follows from the definition that 
$\|f(v) - f(u)\|_\infty \le 1$ if and only if $u$ and $v$ are adjacent in $G$.
Thus, cubicity of a graph $G$ is also the minimum dimension $d$ such that there exist a function 
$f : V(G) \mapsto {\mathbb{R}}^d$ such that $\|f(v) - f(u)\|_\infty \le 1$ if and only if $u$ and $v$ are adjacent in $G$.

Now we will turn our attention to the special case of trees and show that there is a correspondence between the maps 
from $V(T)$ to ${\mathbb{R}}^d$ as discussed above, and 
weight-vector assignments to edges $E(T) \mapsto [-1, 1]^d$ with some nice properties. 
Let $r$ denote an arbitrarily chosen root vertex of $T$ and let $h$ be the height of the rooted tree $T$. 
Suppose we have a weight-vector assignment $W : E(T) \mapsto [-1, 1]^d$. 
For any vertex $v \ne r$, let $S_{W}(v)$ be the sum of weight-vectors of edges 
along the path in $T$ from $r$ to $v$, under the weight-vector assignment $W$ and let $S_{W}(r)$ be the zero vector. 
Note that if $u$ and $v$ are adjacent in $T$, then $\|S_W(u) - S_W(v)\|_\infty \le 1$.
\begin{definition}
 Let $W$ be a weight-vector assignment such that $W: E(T) \mapsto [-1, 1]^d$ and $S_W$ be defined with respect to $W$, as above.  
We say that $W$ is a separating weight-vector assignment for a pair $u, v$ of non-adjacent vertices of $T$, 
if $\|S_W(u) - S_W(v)\|_\infty > 1$.
\end{definition}
\textit{If $W: E(T) \mapsto [-1, 1]^d$ is a separating weight-vector assignment for every pair $u, v$ of non-adjacent vertices of $T$, 
then the function $f: V(T) \mapsto {\mathbb{R}}^d$ defined as $f(v)=S_W(v)$ corresponds to a $d$-dimensional 
cube representation of $T$.}   

Conversely, given $f : V(T) \mapsto {\mathbb{R}}^d$ such that $\|f(v) - f(u)\|_\infty \le 1$ if and only if $u$ and $v$ are adjacent in $G$, 
we can also get a corresponding weight-vector assignment $W : E(T) \mapsto [-1, 1]^d$ such that 
$\|S_{W}(u) - S_{W}(v)\|_\infty > 1$, if and only if $u$ and $v$ are non-adjacent. 
If $uv \in E(T)$ such that $u$ is the child vertex of $v$, then define $W(uv)=f(u)-f(v)$, 
which will be a vector belonging to $[-1, +1]^d$. From this, it is immediate that whenever 
$u$ and $v$ are adjacent, $\|S_{W}(u) - S_{W}(v)\|_\infty \le 1$.  
If $u$ and $v$ are non-adjacent vertices, we had $\|f(u)-f(v)\|_\infty > 1$. 
Suppose $a$ is the least common ancestor of $u$ and $v$ in $T$ and 
$u=v_0, v_1, v_2, \ldots, v_{j-1}, a=v_j, v_{j+1}, v_k, v_{k+1}=v$ is the path in $T$ between $u$ and $v$. Since
the path from $v_j=a$ to the root vertex is common to both the path from $u$ to $r$ and $v$ to $r$, it is easy to see that 
$S_W(u) - S_W(v) = W(v_0,v_1)+W(v_1, v_2)+\cdots+W(v_{j-1}, v_j)-W(v_j, v_{j+1}) - \cdots - W(v_k, v)$. Therefore, 
$\|S_W(u) - S_W(v)\|_\infty = \|W(u,v_1)+W(v_1, v_2)+\cdots+W(v_{j-1}, v_j)-W(v_j, v_{j+1}) - \cdots - W(v_k, v)\|_\infty$.
Since for any edge $(v_i, v_{i+1})$ in the $uv$ path $W(v_i,v_{i+1})=f(v_i)-f(v_{i+1})$, the RHS is equal to $\|f(u)-f(v)\|_\infty >1$.
We note down the following simple property, since it is used in later parts of the paper as well.
\begin{property}\label{propSumofWeights}
Let $T$ be a tree and $W : E(T) \mapsto [-1, 1]^d$ and for any vertex $v$, let $S_{W}(v)$ be the sum of weight-vectors on the edges along
the path in $T$ from the root of $t$ to $v$, under the weight-vector assignment $W$.
Suppose $u=v_0, v_1, v_2, \ldots, v_k, v_{k+1}=v$ is the path in $T$ between $u$ and $v$. Then,\\
$S_W(u) - S_W(v) = W(u,v_1)+W(v_1, v_2)+\cdots+W(v_{j-1}, v_j)-W(v_j, v_{j+1}) - \cdots -W(v_{k-1}, v_k) - W(v_k, v)$, 
where $v_j$ is the least common ancestor of $u$ and $v$ in $T$.
\end{property}
Our discussion is summarized below:
\begin{lemma}\label{lemCorrespondance}
Given a cube representations of $T$ of dimension $d$, in polynomial time we can compute weight-vector assignment $W: E(T) \mapsto [-1, 1]^d$
that is a separating weight-vector assignment for every pair of non-adjacent vertices $u$ and $v$ of $T$. 
Conversely, given weight-vector assignment $W: E(T) \mapsto [-1, 1]^d$
that is a separating weight-vector assignment for every pair of non-adjacent vertices $u$ and $v$ of $T$, then in polynomial time, we can obtain 
a $d$-dimensional cube representation of $T$.
\end{lemma}
\subsection{Bounds for cubicity}
In this section we discuss some lower bounds and upper bounds for cubicity of tress and combine them to obtain 
cube representations of small dimension for trees having relatively small height. The following is a well known lower bound for the cubicity of general graphs. 
\begin{lemma}[\cite{ChandranMannino}] \label{lbgeneralgraphs}
If $G$ is a graph of diameter $d>0$, on $n$ vertices, then $\operatorname{cub}(G) \ge \left \lceil \frac{\log \alpha(G)}{\log(d+1)} \right \rceil$, 
where $\alpha(G)$ is the cardinality of 
a maximum independent set in $G$.
\end{lemma}
\begin{proof}
Suppose $\operatorname{cub}(G)=k$. This means that $G$ can be represented as the intersection graph of axis parallel hypercubes in $k$ dimensions. 
This cube representation, when projected to the $k$ fundamental directions, give $k$ unit interval supergraphs of $G$, say $I_1, I_2, \ldots, I_k$. 
Clearly, each $I_i$, $1 \le i \le k$ has diameter at most $d$ and in any interval representation of $I_i$, 
the distance between the left end point of the left most unit interval and the right end point of the rightmost unit interval is at most $d+1$.   
This implies that the total volume occupied by the cube representation, in the $k$-dimensional Euclidean space is at most $(d+1)^k$. 
But we know that there are $\alpha(G)$ vertices such that unit volume hypercubes corresponding to no two of them share a common point. 
Therefore, the volume occupied by the cube representation is at least $\alpha(G)$ units. Thus we have, $(d+1)^k \ge \alpha(G)$.
~\hfill$\qed$
\end{proof}
\begin{definition}
Let $G$ be a graph of diameter $d$ and for each $1 \le r \le d$ and $v \in V(T)$, let $B_{v,r}$ represent the set of vertices in $G$, 
which are at a distance at most $r$ from $v$. Then, we define $\rho(G)=\displaystyle\max_{v \in V, 1 \le r \le d}{\frac {\log{\frac{|B_{v, r}|}{2}}}{\log{(2r+1)}}}$.
\end{definition}
The following lemma is a direct consequence of the above definition.
\begin{lemma}\label{lemvertexCount}
 For any $v \in V(G)$ and $1 \le r \le diameter(G)$, $|B_{v,r}| \le 2 (2r+1)^{\rho(G)}$.
\end{lemma}
\begin{theorem}\label{lbtrees}
 For any tree $T$, $\operatorname{cub}(T) \ge \lceil \rho(T) \rceil$.
\end{theorem}
\begin{proof}
This directly follows from Lemma \ref{lbgeneralgraphs}, because the subtree of $T$ induced on $B_{v,r}$ has an independent set of 
size at least $\frac{|B_{v,r}|}{2}$ and diameter at most $2r$.
~\hfill$\qed$
\end{proof}
\begin{remark}
 It should be noted that it is only in the case of trees that the parameter $\rho$ is a lower bound for cubicity. 
In the case of general graphs, this is not applicable. An easy counter example would be the case of cliques.
\end{remark}
\begin{lemma}\label{lemtrivupper}
For any tree $T$ on $n$ vertices, $\operatorname{cub}(T) \le 1+ \left \lceil \log n \right \rceil$ and a cube representation of $T$ of 
dimension $1+ \left \lceil \log n \right \rceil$ can be constructed in polynomial time. 
\end{lemma}
\begin{proof}
Shah \cite{Chinthan} describes a polynomial time algorithm for constructing two interval supergraphs $I_1$ and $I_2$ of $T$
such that $V(T)=V(I_1)=V(I_2)$, $I_1$ is a unit interval graph and $E(T)=E(I_1) \cap E(I_2)$. Since we also know that
any interval graph has $\left \lceil \log n \right \rceil$-dimensional cube representation and in polynomial time 
we can construct $\left \lceil \log n \right \rceil$ unit interval graphs on the same vertex set $V(T)=V(I_2)$ such that the intersection 
of their edge sets is $E(I_2)$ \cite{Chandran09}. From this, the statement follows.
\qed
\end{proof}
\begin{lemma}\label{lemSmallTreeCubicity}
 Let $T$ be a tree with $\operatorname{cub}(T) \ge 2$ and $T_i$ be a subtree of $T$ of height at most $2^{2^4}$. 
 Then, a cube representation of $T_i$ of dimension $\lceil c \times \rho(T) \rceil + 2  \le (c+1) \times \operatorname{cub}(T)$ or more can be 
 constructed in polynomial time, where $c=22.77$. 
\end{lemma}
\begin{proof}
If $\operatorname{cub}(T) \le 1$, $T$ should be path; otherwise, it has an induced star on four vertices, denoted as $K_{1, 3}$, 
which forces $\operatorname{cub}(T) \ge 2$ \cite{Chandran09}. 
Since we assumed that $\operatorname{cub}(T) \ge 2$, $T$ contains an induced $K_{1, 3}$ and therefore, 
$\lceil \rho(T) \rceil \ge 1$. 
If $\operatorname{cub}(T_i) \ge 2$, by Lemma \ref{lemvertexCount},  
$|V(T_i)| \le 2 (2^{17}+1)^{\rho(T)}$. By Lemma \ref{lemtrivupper}, a cube representation of $T$ 
of dimension $d \le 2+ \left \lceil \rho(T) \log (2^{17}+1)\right \rceil$ can be constructed in polynomial time. 

After getting a cube representation of $T_i$ in a lower dimension $d_1$, it is a trivial job to extend it to a higher dimension $d_2$. 
Consider the cube representation as a mapping $f : V(T) \mapsto {\mathbb{R}}^{d_1}$, as described in Section \ref{weightAssignments} and for each $v \in V(T)$, append 
the vector $f(v)$ with $d_2- d_1$ additional coordinates each of whose value is zero. 
By Lemma \ref{lbtrees}, the statement follows.
~\hfill$\qed$
\end{proof}
\section{Constructing the cube representation}
Only cliques have cubicity zero. 
If a tree has a vertex of degree three, its cubicity is greater than one, since it has an induced $K_{1, 3}$ \cite{Chandran09}. 
Therefore, a tree of 
cubicity one can be only a path, whose unit interval representation is easy to construct. 
Hence, for the remaining parts of this paper, we assume that $\operatorname{cub}(T) \ge 2$. 
This also means that $n \ge 4$ and $\lceil \rho(T) \rceil \ge 1$. 

In the previous section, we saw that for a tree $T$, $\lceil \rho(T) \rceil$ is a lower bound for $\operatorname{cub}(T)$. 
Since $\rho(T)$ can be computed in polynomial time by its definition, if we can show the existence of a constant $c$ 
such that $\operatorname{cub}(T) \le c \lceil \rho(T) \rceil$ for any tree $T$, then $c \lceil \rho(T) \rceil$ will serve as 
a polynomial time computable $c$ factor approximation for $\operatorname{cub}(T)$. The existence and determination of such a constant is proved 
using probabilistic arguments and the techniques we describe below are essentially derived from the 
techniques used in Krauthgamer et al. \cite{Krauth2007}. The method also gives a randomized algorithm to compute 
the corresponding cube representation. 
\subsection{A recursive decomposition of trees}
We first define a recursive decomposition of the rooted tree $T$ into rooted subtrees. 

Let $h$ denote the height of the tree $T$. 
Let $k=\lceil \log\log h \rceil$ and $\Gamma = 2^{2^k}$. 
Clearly, $\sqrt{\Gamma} = 2^{2^{k-1}} < h \le 2^{2^k} = \Gamma$. 
For each $0 \le i \le k-1$, let $h_i=\Gamma^{\frac{1}{2^i}}$. 
Thus, $h_0=\Gamma$ and $h_{i+1} = \sqrt{h_i}$. 
Let $e$ denote the minimum even integer such that $h_e \le 2^{16}$ and $o$ denote the minimum odd 
integer such that $h_o \le 2^{16}$. (This means $\{h_e, h_o\} = \{2^{2^3}, 2^{2^4}\}$). 

For each integer $i$ such that $\max(e, o) \ge i \ge 0$ we define two sets of rooted subtrees of $T$ as follows: 
If we delete all edges of $T$ that connect vertices at depth $j$ and $j+1$ for each $j$ which is a positive integer multiple of $3h_i$, 
the tree $T$ gets decomposed into several vertex disjoint subtrees. We consider each such subtree as a rooted subtree with its root being the
vertex in the subtree of smallest depth with respect to $T$. We denote this family of rooted subtrees of $T$ as $\mathcal{A}_i$.
In a similar way, let $\mathcal{B}_i$ denote the family of rooted subtrees of $T$, 
obtained by deleting all edges of $T$ that connect vertices at depth $j$ and $j+1$ for each $j$ such that $j \equiv h_i \mod 3h_i$. 
Let $O_i^A$ denote the set of edges deleted from $T$ to form $\mathcal{A}_i$ and let $O_i^B$ denote the set of edges deleted from $T$ to form $\mathcal{B}_i$. 
Let $\mathcal{L}_i=\mathcal{A}_i \cup \mathcal{B}_i$. 
\begin{lemma} \label{lemSeparationLevel}
For each $i$ such that $\max(e, o) \ge i \ge 0$:
\begin{enumerate}
  \item The rooted trees in $\mathcal{L}_i$ have height at most $3h_i$. 
 \item Trees in $\mathcal{A}_{i+1}$ are subtrees of trees in $\mathcal{A}_{i}$ and trees in $\mathcal{B}_{i+1}$ are subtrees of trees in $\mathcal{B}_i$. 
This is because $O_i^A \subseteq O_{i+1}^A$ and $O_i^B \subseteq O_{i+1}^B$. 
 \item Vertex sets of trees in $\mathcal{A}_i$ partition $V(T)$. Same is the case with $\mathcal{B}_i$s. 
  \item If $u$ and $v$ are two vertices such that $d_{uv} \le h_i$, then there exist at least one subtree $F \in \mathcal{L}_i$ 
 such that both $u$ and $v$ belong to $V(F)$. 
\end{enumerate}
\end{lemma}
\begin{proof}
The first three parts of the lemma follow directly from the definitions. Here we will prove the last part of the lemma.

 Assume that $d_{uv} \le h_i$ in $T$ and let $x$ be the least common ancestor of $u$ and $v$ in $T$. 
Without loss of generality, let $d_{vx} \le d_{ux}\le h_i$. 
Let $F$ be the tree in $\mathcal{A}_i$ such that $x \in V(F)$ and let $r$ be the root of $F$. 
We know that $d_{rx} \le 3 h_i$, by construction of $F$.
If $d_{rx} \le 2h_i$, then $d_{rv} \le d_{ru} = d_{rx}+ d_{xu} \le 2h_i + h_i \le 3 h_i$ and therefore, $v, u \in V(F)$, by construction. 

On the other hand, if $d_{rx} > 2 h_i$, then $\exists y \in V(F)$ such that $d_{ry}=h_i+1$ and $y$ is on the path from $r$ to $x$ in $T$. 
By our construction, $y$ becomes the root of a tree $F' \in \mathcal{B}_i$. 
Since $d_{rx} \le 3 h_i$ by construction of $F$ and $d_{ry}=h+1$, we have $d_{xy}=d_{rx}-d_{ry} < 2 h_i$. 
This gives $d_{uy}=d_{ux}+d_{xy} < h_i + 2 h_i = 3 h_i$ Similarly, $d_{vy}=d_{vx}+d_{vy} < h_i + 2 h_i = 3 h_i$. Therefore, $v, u \in V(F')$, by construction. 
 ~\hfill$\qed$
\end{proof}
\begin{definition}\label{defCombineWeights}
 If $T_1, T_2, \ldots, T_k$ are trees with disjoint vertex sets and for $1 \le j \le k$, $W_j: E(T_j) \mapsto [-1, 1]^d$, then 
 a weight-vector assignment $W: E(T_1) \cup E(T_2) \cup \cdots \cup E(T_k) \mapsto [-1, 1]^d$ can be obtained by assigning $W(e)=W_j(e)$, where 
 $T_j$ is the tree containing the edge $e$. Then, $W$ is the weight-vector assignment for $E(T_1) \cup E(T_2) \cup \cdots \cup E(T_k)$ 
 derived from $W_1, W_2, \ldots, W_k$. 
\end{definition}
\subsection{A randomized algorithm for constructing the cube representation}
From our definitions, $\{h_e, h_o\} = \{2^{2^3}, 2^{2^4}\}$. 
The idea of recursive decomposition of trees and extending the weight-vector assignments of smaller trees to  weight-vector assignments of bigger trees 
was used by Krauthgamer et al. \cite{Krauth2007} 
to attain injectivity while embedding the vertices in $Z^d_\infty$. 
As we will explain soon, the same technique helps us to make sure that the hypercubes corresponding to non-adjacent vertex pairs do not intersect. 
The algorithm for constructing a weight-vector assignment for $E(T)$ that separates every pair of non-adjacent vertices of $T$ is given below:
\begin{enumerate}
 \item Using Lemma \ref{lemSmallTreeCubicity}, construct cube representations of dimension $t=\lceil 22.77 \times \rho(T) \rceil +2$ for each of the subtrees 
belonging to $\mathcal{L}_{e} \cup \mathcal{L}_{o}$.
 \item Using the correspondence given in Section \ref{weightAssignments} between cube representations and weight-vector assignments, 
for each tree $F \in \mathcal{A}_e \cup \mathcal{B}_e \cup \mathcal{A}_o \cup \mathcal{B}_o$, compute a weight-vector assignment $W_e^F: E(F) \mapsto [-1, 1]^{t}$. 
Notice that $\bigcup_{F \in \mathcal{A}_e}{E(F)}=E(T) \setminus O_{e}^A$.
Combine the weight-vector assignments of trees in $\mathcal{A}_e$ as in Definition \ref{defCombineWeights} and obtain $W^A_e : E(T) \setminus O_{e}^A \mapsto [-1, 1]^{t}$. Similarly, obtain $W^B_e : E(T) \setminus O_{e}^B \mapsto [-1, 1]^{t}$ from weight-vector assignments of trees in 
$\mathcal{B}_e$, $W^A_o : E(T) \setminus O_{o}^A \mapsto [-1, 1]^{t}$ from weight-vector assignments of trees $F \in \mathcal{A}_o$ 
and $W^B_o : E(T) \setminus O_{o}^B \mapsto [-1, 1]^{t}$ from weight-vector assignments of trees in $\mathcal{B}_o$.
\item Set $i=\max(e, o)$ and repeat steps 3a to 3d while $i>1$.
\begin{enumerate}
\item For each edge $uv$ belonging to $E(T) \setminus O_{i}^A$, assign $W^A_{i-2}(uv)=W^A_{i}(uv)$ and for each edge $uv$ belonging to $E(T) \setminus O_{i}^B$, assign $W^B_{i-2}(uv)=W^B_{i}(uv)$.
\item For each tree $F \in \mathcal{A}_{i-2}$, do the following: For each edge $uv$ of $F$ such that $uv \in O_{i}^A \setminus O_{i-2}^A$, 
$W^A_{i-2}(uv)$ is assigned a weight-vector from $\{-1, 1\}^t$, chosen uniformly at random. Now, each edge $uv$ of $F$ has got a weight-vector under $W^A_{i-2}$. 
For each vertex $v$ of $F$, compute $S(v)$ as the sum of weight-vectors on edges of the path in $F$ from the root of $F$ to $v$, as given by $W^A_{i-2}$. 
For each pair of non-adjacent vertices $u$ and $v$ of $F$ such that $d_{uv} \ge h_{i-1}$, check whether $\|S(v) - S(u)\|_\infty > 1$. 
Repeat Step 3b, until the above condition becomes true simultaneously for all pair of non-adjacent vertices $u$ and $v$ of $F$ such that $d_{uv} \ge h_{i-1}$.
\item For each tree $F \in \mathcal{B}_{i-2}$, do the following: For each edge $uv$ of $F$ such that $uv \in O_{i}^B \setminus O_{i-2}^B$, 
$W^B_{i-2}(uv)$ is assigned a weight-vector from $\{-1, 1\}^t$, chosen uniformly at random. Now, each edge $uv$ of $F$ has got a weight-vector under $W^B_{i-2}$. 
For each vertex $v$ of $F$, compute $S(v)$ as the sum of weight-vectors on edges of the path in $F$ from the root of $F$ to $v$, as given by $W^B_{i-2}$. 
For each pair of non-adjacent vertices $u$ and $v$ of $F$ such that $d_{uv} \ge h_{i-1}$, check whether $\|S(v) - S(u)\|_\infty > 1$. 
Repeat Step 3c, until the above condition becomes true simultaneously for all pair of non-adjacent vertices $u$ and $v$ of $F$ such that $d_{uv} \ge h_{i-1}$.
\item Set $i=i-1$.
\end{enumerate}
\item For each edge $uv$ belonging to $E(T) \setminus O_{1}^A$, assign $W'^A_{0}(uv)=W^A_{1}(uv)$ and for each edge $uv$ belonging to $O_{1}^A$, 
assign the all zeros vector to $W'^A_{0}(uv)$. Similarly, for each edge $uv$ belonging to $E(T) \setminus O_{1}^B$, assign $W'^B_{0}(uv)=W^B_{1}(uv)$ 
and for each edge $uv$ belonging to $O_{1}^B$, assign the all zeros vector to $W'^B_{0}(uv)$. 
\item Output $W^A_{0} \circ W^B_{0} \circ W'^A_{0} \circ W'^B_{0}$, a weight-vector assignment from $E(T)$ to $[-1, 1]^{4t}$ obtained by concatenating 
the components of weight assignments $W^A_{0}$, $W^B_{0}$, $W'^A_{0}$ and $W'^B_{0}$ together. 
\end{enumerate}
\begin{property}\label{propSeparation}
 $W^A_{0} \circ W^B_{0}$ is a separating weight-vector assignment for every non-adjacent pair of vertices $u$ and $v$ of $T$ such that $d_{uv} \le h_e$ 
or $h_{i-1} \le d_{uv} \le  h_{i-2}$, for any even integer $i$ such that $e \ge i \ge 2$. 
Similarly, $W'^A_{0} \circ W'^B_{0}$ is a separating weight-vector assignment for every non-adjacent pair of vertices $u$ and $v$ of $T$ such that $d_{uv} \le  h_o$ 
or $h_{i-1} \le d_{uv} \le h_{i-2}$, for any odd integer $i$ such that $o \ge i \ge 3$. 
\end{property}
\begin{proof}
Let $u$ and $v$ be two non-adjacent vertices in $T$. 
If $d_{uv} \le h_e$, by part 4 of Lemma \ref{lemSeparationLevel}, there exist at least one subtree $F \in \mathcal{L}_e$ 
such that both $u$ and $v$ belong to $V(F)$. In step 2 of the algorithm, we computed $W_e^F$ from a cube representation of $F$, 
which is a separating weight-vector assignment by the correspondence given in Lemma \ref{lemCorrespondance}. If $F \in \mathcal{A}_e$,
then $W_e^F$ is one of the weight-vector assignment from which $W^A_e$ is derived, and on each edge of the path from $u$ to $v$ in $T$,
the weight-vector assigned by $W^A_e$ is the same as the weight-vector assigned by $W_e^F$. Since each edge $xy$ of the path from $u$ to $v$ 
in $T$ belongs to $E(T) \setminus O_{e}^A$, in Step 3a the algorithm assigns $W^A_{i-2}(xy)=W^A_{i}(xy)$ for each even integer $i$ where $e \ge i \ge 2$.
Thus, finally we will have $W^A_{0}(xy)=W^A_{e}(xy)=W_e^F(xy)$. Therefore,   
by Property \ref{propSumofWeights} it follows that $W^A_0$ will be a separating weight-vector assignment
for $u$ and $v$. By similar reasons, if $F \in \mathcal{B}_e$, $W^B_0$ will be a separating weight-vector assignment
for $u$ and $v$. 

Similarly, if $h_{i-1} \le d_{uv} \le  h_{i-2}$, for any even integer $i$ such that $e \ge i \ge 2$, then 
by part 4 of Lemma \ref{lemSeparationLevel} there exist at least one subtree $F \in \mathcal{L}_{i-2}$ such that 
both $u$ and $v$ belong to $V(F)$. 
If $F \in \mathcal{A}_{i-2}$, in step 3a of the algorithm we would have made sure that $W^A_{i-2}$ is a 
separating weight-vector assignment for $u$ and $v$. As in the earlier case, for each edge $xy$ of the path from $u$ to $v$ 
in $T$, $W^A_{0}(xy)=W^A_{i-2}(xy)$ and by Property \ref{propSumofWeights}, $W^A_0$ will be a separating weight-vector assignment
for $u$ and $v$. Similarly, if $F \in \mathcal{B}_{i-2}$, $W^B_0$ will be a separating weight-vector assignment
for $u$ and $v$. 

Thus, for every non-adjacent pair of vertices $u$ and $v$ of $T$ such that $d_{uv} \le h_e$ 
or $h_{i-1} \le d_{uv} \le  h_{i-2}$ for any even integer $i$ such that $e \ge i \ge 2$ one of $W^A_{0}$ and $W^B_{0}$ 
is a separating weight-vector assignment, which implies that $W^A_{0} \circ W^B_{0}$ is a separating weight-vector assignment
for $u$ and $v$. 

The proof of the second part of the lemma is similar. If $u$ and $v$ are non-adjacent pairs of vertices of $T$ such that 
$d_{uv} \le  h_o$ or $h_{i-1} \le d_{uv} \le h_{i-2}$, for any odd integer $i$ such that $o \ge i \ge 3$, 
then there exist at least one subtree $F \in \mathcal{L}_{i-2}$ such that 
both $u$ and $v$ belong to $V(F)$. If $F \in \mathcal{A}_{i-2}$, we get $W'^A_{0}(xy)=W^A_{1}(xy)=W^A_{i-2}(xy)$ and if 
$F \in \mathcal{B}_{i-2}$, we get $W'^B_{0}(xy)=W^B_{1}(xy)=W^B_{i-2}(xy)$, for each edge $xy$ of the path from $u$ to $v$ 
in $T$. This implies that 
$W'^A_{0} \circ W'^B_{0}$ is a separating weight-vector assignment
for $u$ and $v$. 
  ~\hfill$\qed$
\end{proof}
The following is a direct consequence of Property \ref{propSeparation}. 
\begin{theorem}
$\mathcal{W}= W^A_{0} \circ W^B_{0} \circ W'^A_{0} \circ W'^B_{0}$ is a separating weight-vector assignment for each non-adjacent pair of vertices $u$ and $v$ of $T$. 
Here, $W: E(T) \mapsto [-1, 1]^{4t}$, where $t=\lceil 22.77 \times \rho(T) \rceil +2$.
\end{theorem}
The following lemma will help us to calculate the expected number of times the algorithm repeats Step 3b (or 3c) till it obtains a suitable weight-vector assignment 
for a tree $F \in \mathcal{L}_{i-2}$, where $e \ge i \ge 2$. 
\begin{lemma}\label{extendWtOnce}
Let $i$ be such that $h_i \ge 2^{2^3}$ and $i \ge 2$. Let $W_i: E(T) \setminus O_{i}^A \mapsto [-1, 1]^t$, where 
$t= \lceil 22.77 \times \rho(T) \rceil +2$ and $F \in \mathcal{A}_{i-2}$. 
Suppose for each edge $uv$ of $F$ such that $uv \in E(T) \setminus O_{i}^A$, 
we set $W_{i-2}(uv)=W_{i}(uv)$ and for each edge $uv$ of $F$ such that $uv \in O_{i}^A \setminus O_{i-2}^A$, we assign $W_{i-2}(uv)$ to be a vector 
from $\{-1, 1\}^t$ chosen independently and uniformly at random.  
For each vertex $v$ of $F$, let $S_{W_{i-2}}(v)$ be the sum of 
edge weights of the edges belonging to the path from the root of $F$ to $v$, as given by $W_{i-2}$.
Then, with probability at least $p=0.64$, for every pair of non-adjacent vertices $u$ and $v$ of $F$ such that 
$d_{uv} \ge h_{i-1}$, $\|S_{W_{i-2}}(v) - S_{W_{i-2}}(u)\|_\infty > 1$. 
\end{lemma}
\begin{proof}
Consider a pair of non-adjacent vertices $u$ and $v$ belonging to the vertex set of 
the same rooted subtree $F \in \mathcal{A}_{i-2}$ and $d_{uv} \ge h_{i-1}$. Let $r$ be the root of $F$. 
Since $u$ and $v$ both belong to the same subtree $F \in \mathcal{A}_{i-2}$, all the edges in the $uv$ path fall in $E(T) \setminus O_{i-2}^A$. 
Therefore, all the edges in the $uv$ path get their weight-vectors assigned under $W_{i-2}$. 
But since $d_{uv} \ge h_{i-1} = {h_{i}}^2$ and $h_i \ge 2^8$ and each subtree in $\mathcal{A}_i$ has height at most $3 h_{i}$, among the edges in the $uv$ path, 
at least $\frac{h_i}{4}$ edges should belong to $O_{i}^A \setminus O_{i-2}^A$ and got their weights assigned 
independently and uniformly at random from $\{-1, 1\}^t$, as stated in the lemma. 
The other edges on the $uv$ path were already assigned values in $W_i$ and these values remain the same in $W_{i-2}$. 
Let $u=v_0, v_1, v_2, \ldots, v_q, v_{q+1}=v$ be the path in $T$ between $u$ and $v$, where $v_j$ is the least common ancestor of $u$ and $v$ in $T$.
Also let $S^k$ denote the $k^{th}$ coordinate function of $S_{W_{i-2}}$ and $W^k$ denote the $k^{th}$ coordinate function of $W_{i-2}$.  
By property \ref{propSumofWeights}, for each $1 \le k \le t$, 
$S^k(u) - S^k(v) =  X_k + c_k$, where $X_k =\sum_{\{0 \le i \le j-1 \text{ and } v_i v_{i+1} \in O_{i-2}^A \setminus O_{i}^A \}}{W^k(xy)} - \sum_{\{j \le i \le q \text{ and } v_i v_{i+1} \in O_{i-2}^A \setminus O_{i}^A \}}{W^k(xy)}$ and $c_k \in {\mathbb{R}}$ is a constant, depending on the weight vectors fixed by $W_i$ for edges in the $uv$ path that belong to $E(T) \setminus O_{i}^A$. Let $l$ be the number of edges in the $uv$ path that belong to $E(T) \setminus O_{i}^A$.

We will bound the probability that $|S^k(v) - S^k(u)| \le 1$. Note that $X_k$ is the sum of $l$ iid random variables, each of which is $-1$ or $+1$ with equal probability. 
Therefore,  
$$Pr(|S^k(v) - S^k(u)| \le 1)= Pr(X_k \text{ falls in the interval }[-c_1-1, -c_1+1]$$
But since $X_k$ can take only integer values and $X_k$ can take at most two possible values in $[-c_1-1, -c_1+1]$ irrespective of whether $l$ is even or odd, because any interval of length two can contain at most two integers of the same parity. 
Therefore, $Pr(|S^k(v) - S^k(u)| \le 1) \le 2 {l \choose \left\lceil \frac{l}{2}\right \rceil} 2^{-l}$.
Since $l \ge \frac{h_i}{4}\ge 2^6$, using Sterling's approximation formula, 
$$Pr(|S^k(v) - S^k(u)| \le 1) \le \frac{1.61}{\sqrt{l}} \le \frac{1.61}{\sqrt{\frac{h_i}{4}}}$$
$$Pr(\|S_{W_{i-2}}(v) - S_{W_{i-2}}(u)\|_\infty \le 1) \le \left(\frac{1.61}{\sqrt{\frac{h_i}{4}}}\right)^{t}$$
Since the height of $F$ is at most $3h_{i-2}$, by Lemma \ref{lemvertexCount}, there are at most $2(6 h_{i-2} + 1)^{\rho(T)}$ vertices in $T_i$ and 
the number of non-adjacent pairs $u, v \in V(F)$ such that 
$h_{i-1} \le d_{uv} \le 2 \times h_{i-1}$,  is at most 
$4(6 h_{i-2} + 1)^{\rho(T)} \left (2 \times 2h_{i-1}+1\right)^{\rho(T)}$. 

For each integer $l$ where $1\le l \le \log(h_{i-1})$, let $\mathcal{P}_l$ denote the set consisting
of the non-adjacenct pairs $u, v \in V (F)$ such that $2^{l-1}h_{i-1} \le d_{uv} \le 2^l h_{i-1}$. 
Using Lemma \ref{lemvertexCount}, it is easy to see that for each integer $l$ where $1 \le l \le \log(h_{i-1})$,
$|\mathcal{P}_l| \le 4(6h_{i-2} + 1)^{\rho(T)} (2^l 2h_{i-1} + 1)^{\rho(T)}$. Using similar arguments as given in
the previous paragraph, we also get the following:
For each pair $(u, v) \in \mathcal{P}_l$,
$$Pr(\|S_{W_{i-2}}(v) - S_{W_{i-2}}(u)\|_\infty \le 1) \le \left(\frac{1.61}{\sqrt{(2^{l-1}\times \frac{h_i}{4} )}}\right)^{t}$$
Applying union bound,
$$Pr(\exists u, v \in V(F) \text{ with } d_{uv} \ge h_{i-1}\text{ and }\|S_{W_{i-2}}(v) - S_{W_{i-2}}(u)\|_\infty \le 1)$$
\begin{align*}
&\le \sum_{l=1}^{\log(h_{i-1})}{ |\mathcal{P}_l| \left(\frac{1.61}{\sqrt{(2^{l-1} \times \frac{h_i}{4})}}\right)^{t}}\\
&\le \sum_{l=1}^{\log(h_{i-1})}{4(6 h_{i-2} + 1)^{\rho(T)} \left(2^{l} 2 h_{i-1}+1\right)^{\rho(T)} \left(\frac{1.61}{\sqrt{(2^{l-1} \times \frac{h_i}{4})}}\right)^{t}}\\
&\le 8(6 h_{i-2} + 1)^{\rho(T)} \left(2 \times 2 h_{i-1}+1\right)^{\rho(T)} \left(\frac{1.61}{\sqrt{\frac{h_i}{4}}}\right)^{t}\\
&\le 0.33, \text{ since $t\ge \lceil 22.77 \times \rho(T) \rceil +2$, $h_i \ge 2^{2^3}$ and $h_{i-2}=h_{i-1}^2=h_i^4.$}
\end{align*}
 Therefore, with probability at least $0.67$, for every pair of non-adjacent vertices $u$ and $v$ of $F$ such that 
$d_{uv} \ge h_{i-1}$, $|S_{w_i}(v) - S_{w_i}(u)| > 1$ for some $k$ such that $1 \le k \le t$.
~\hfill$\qed$
\end{proof}
\begin{lemma}\label{lemRepeatCount}
The expected number of times the algorithm repeats Step 3b (or 3c) till it obtains a suitable weight-vector assignment for a tree $F \in \mathcal{L}_{i-2}$ is at most $\frac{1}{0.67}$ for any $i$ such that $e \ge i \ge 2$.
\end{lemma}
\begin{theorem}
For any $T$, we can compute a $4(\lceil 22.77 \times \rho(T) \rceil +2)$-dimensional cube representation using a randomized algorithm which runs in time polynomial in expectation. Cubicity of trees can be approximated within a constant factor in deterministic polynomial time.
\end{theorem}
\begin{proof}
The second part of the theorem follows from the first part, because $\rho(T)$ is a polynomial time computable function. 
Since by Lemma \ref{lemCorrespondance}, in polynomial time we can construct a $d$-dimensional cube representation of $T$ from a weight-vector assignment 
$W : E(T) \mapsto [-1, 1]^d$, it is enough to show that the randomized algorithm we described here, for computing a weight-vector assignment $W : E(T) \mapsto [-1, 1]^{4t}$, where $t= \lceil 22.77 \times \rho(T) \rceil +2$ runs in time polynomial in expectation. 

In any partition of the rooted tree $T$ into smaller trees, there can be at most $O(n)$ rooted subtrees. Therefore, by Lemma \ref{lemSmallTreeCubicity}, step 1 of 
the algorithm runs in polynomial time. In step 2 of the algorithm, the weight-vector assignments can be computed in  polynomial time, by Lemma \ref{lemCorrespondance}.
The operation in step 2 of combining the weight assignments on smaller trees as given in Definition \ref{defCombineWeights} can easily be done in polynomial time. 
By the definition of the recursive decomposition, Step 3 is executed at most $O(\log \log h)$ rounds, where $h$ is the height of the tree $T$. 
It is easy to see that the assignments in step 3a can be done in polynomial time. By Lemma \ref{lemRepeatCount}, for each round of execution of step 3, steps 3b and 3c
are repeated only constantly many times in expectation. In each repetition, the algorithm does only a polynomial time operation. 
Steps 4 and 5 are simple assignments, which can be done in polynomial time. 
 ~\hfill$\qed$
\end{proof}
\section{Conclusions}
In this paper, we show that cubicity of trees can be approximated within a constant factor, in deterministic polynomial time. 
As far as we know, this is the first constant factor approximation algorithm known for cubicity of trees. 
A corresponding cube representation of the tree can also be computed by a randomized algorithm which runs 
in time polynomial in expectation. 
The basic techniques for the randomized algorithm are borrowed from the techniques given by 
Krauthgamer et al. \cite{Krauth2007}, for approximating the intrinsic dimensionality of trees.
We feel that this is a surprising coincidence because as we have explained in Appendix, 
intrinsic dimensionality is quite different from cubicity and neither the bounds of these parameters nor 
the proof techniques for these problems work for each other in general. As far as we know, till now there are no works connecting 
the parameters cubicity and intrinsic dimension.

\newpage
\appendix
\section{Appendix}
\subsection{Cubicity and intrinsic dimensionality}\label{paramCompare}
Let $Z$ denote the set of integers and $Z^d_\infty$ be the infinite graph with vertex set 
$Z^d$ and an edge $(u, v)$ for two vertices $u$ and $v$ if and only if $\|u - v\|_\infty = 1$. 
The intrinsic dimensionality of a graph $G$, $\operatorname{dim}(G)$ is the smallest $d$ such that $G$ can be 
injectively embedded on to $Z^d_\infty$. This means that $G$ occurs as a (not necessarily induced) subgraph of $Z^d_\infty$. 

Though both cubicity and intrinsic dimensionality are parameters related to graph embeddings, there are several fundamental differences between them.\\
(1) \textbf{Injectivity:} Intrinsic dimensionality requires the mapping from $V(G)$ to $Z^d$ to be injective. Thus, dense graphs will have relatively high intrinsic 
dimensionality compared to sparse graphs. A clique on $n$ vertices has intrinsic dimensionality $\log_2 n$. 
In contrast, the injectivity constraint is absent for cubicity. In a cube representation, the hypercubes corresponding to two distinct vertices are 
permitted to occupy the same space. Recall that a clique has cubicity zero.\\
(2) \textbf{Vertex Positioning:} In the case of intrinsic dimensionality, we should map the vertices of the graph to points in $Z^d$. However, as we saw in the previous section, the mappings associated with cubicity are from $V(G)$ to ${\mathbb{R}}^d$, giving us more freedom to place the hypercubes corresponding to the vertices. 
There are some graphs for which there is a cube representation in ${\mathbb{R}}^2$ even when its vertices have their neighborhoods different from each other, forcing an injective embedding from $V(G)$ to ${\mathbb{R}}^2$. For example, we can show that a graph having two cliques on $n$ vertices and a matching connecting the corresponding pairs of vertices in both the cliques has cubicity two. Thus, even when cube representations have their corresponding vertex embedding injective, cubicity can be very low, due to the flexibility in vertex positioning.\\
(3) \textbf{Treatment of non-adjacency and monotonicity: }In the case of intrinsic dimensionality, it is possible to map even non-adjacent vertices $u$ and $v$ to
points in $Z^d$ which are at unit distance from each other. Because of this freedom, if a graph has intrinsic dimensionality $k$, its subgraphs will have intrinsic 
dimensionality at most $k$. Thus, intrinsic dimensionality is monotone, with respect to subgraph relation. In particular, all graphs of $n$ vertices have intrinsic dimensionality at most that of a clique on $n$ vertices, namely $\log_2 n$.

However, in the case of cube representations, we require the hypercubes corresponding 
to non-adjacent vertices to be non-intersecting and as we discussed in Section \ref{weightAssignments}, the centers of hypercubes of non-adjacent 
vertices are required to be mapped to points in ${\mathbb{R}}^d$ which are at distance strictly more than one. For this reason, the monotonicity we observed in the case of 
intrinsic dimensionality does not happen for cubicity. A clique has cubicity zero, but almost all graphs on $n$ vertices have cubicity $\Omega(n)$ \cite{AdigaLb}.\\ 
(4) \textbf{Parameter value range:} As we noted, almost all graphs on $n$ vertices have cubicity $\Omega(n)$. There are graphs on $n$ vertices with cubicity $\left \lceil\frac{2 n}{3} \right \rceil$. But the intrinsic dimensionality of a graph on $n$ vertices is at most $\log_2 n$.\\
(5) \textbf{Good polynomial time approximations:}
Krauthgamer et al. \cite{Krauth2007} defined a parameter called \textit{growth rate} of a graph $G$, 
defined as $$\eta(G)= sup \left\{ \frac{\log |B_{v,r}|}{\log r} \mid v \in V(G) \text{, }r > 1\right\}$$ 
Note that, this parameter is computable in polynomial time and for a graph on $n$ vertices, the value of this parameter is at most $\log n$.
Krauthgamer et al. \cite{Krauth2007} showed that for any graph $G$, 
its growth rate is a lower bound for its intrinsic dimensionality. They also showed that $dim(G)$ is $O(\eta(G) \log \eta(G))$ in general and 
in the special case of trees, $dim(G)$ is $O(\eta(G))$. This leads to an $O(\log \log n)$ factor approximation algorithm for the 
intrinsic dimensionality of general graphs and a constant factor approximation algorithm in the case of trees. 
For cubicity, the bound given by Lemma \ref{lbgeneralgraphs} is the only
non-trivial polynomial time computable lower bound known. However, notice that this parameter can only go up to $\log_2 n$, whereas 
almost all graphs on $n$ vertices have cubicity is $\Omega(n)$ \cite{AdigaLb}. Moreover, cubicity is known to be inapproximable in polynomial time, 
within an $O(n^{1-\epsilon})$ factor for any $\epsilon >0$, unless $\text{NP}=\text{ZPP}$. 

Thus, cubicity and intrinsic dimensionality are two graph parameters, not directly comparable with each other in general.
There are graphs for which cubicity exceeds intrinsic dimensionality, and for some others it is the other way. 
Even in the special case of trees, the intrinsic dimension and cubicity can be different. For example, 
a star graph $K_{1, n}$ has intrinsic dimension $\log_3 (n+1)$, whereas the same graph has cubicity $\log_2 n$\cite{Linial1995,Chandran09}.  

In spite of all these contrasts between cubicity and intrinsic dimension, they share an interesting similarity: 
The injectivity requirement places a lower bound on the volume required for injectively embedding a graph on to $Z^d_\infty$ and 
this is the reason for having growth rate as a lower bound for intrinsic dimensionality. 
In the case of cubicity, cubes corresponding to non-adjacent pairs of vertices need to be non-intersecting, giving a lower 
bound to the volume required for placing the cubes. This fact was exploited to obtain the lower bounds given by Lemma \ref{lbgeneralgraphs} and 
Theorem \ref{lbtrees}. In a retrospective analysis, it appears that this similarity is what helped
us to use the techniques developed by Krauthgamer et al. \cite{Krauth2007} in developing our algorithm. 
We will be showing that cubicity of a tree $T$ is $O(\rho(T))$. 

However, as we noted under item (5) above, this similarity between the parameters is not powerful enough to be useful 
in the case of general graphs, because of the approximation hardness results. The techniques do not seem to scale up even in 
other special cases, for example, for graphs without long induced simple cycles. Using the result obtained for trees, 
Krauthgamer et al. \cite{Krauth2007} had showed that graphs 
without induced simple cycles of length greater than $\lambda$ have intrinsic dimensionality $O(\eta(G) \log^2(\lambda+2))$. 
But we know that even chordal graphs can have cubicity as high as $\Omega(n)$, whereas the lower bound 
obtained from Lemma \ref{lbgeneralgraphs} can be at most $\log n$. Moreover, even for split graphs which
form a subclass of chordal graphs, cubicity is known to be NP-hard to approximate within an $O(n^{1-\epsilon})$ factor 
for any $\epsilon >0$, unless $\text{NP}=\text{ZPP}$.
\end{document}